\documentclass{amsart}
\usepackage[draft]{optional}

\usepackage[utf8]{inputenc}
\usepackage[english] {babel} 

\usepackage[T1]{fontenc}      
\usepackage{amsmath,amsfonts,amssymb}   
\usepackage{textcomp}         
\usepackage{float}
\usepackage{standalone}
\usepackage{tikz,ifthen}
\usetikzlibrary{decorations.pathmorphing,decorations.pathreplacing,arrows,automata,patterns,positioning}
\usepackage{latexsym}

 \usepackage{graphicx} 

\usepackage{xcomment}
\usepackage{xspace}
\usepackage{cite}
\usepackage{url}
\theoremstyle{plain}
\newtheorem{theorem}{Theorem}[section] 
\newtheorem{lemma}[theorem]{Lemma} 
\newtheorem{corollary}[theorem]{Corollary} 
\theoremstyle{definition}
\theoremstyle{remark}

 \colorlet{colorCarreUpRight}{green!30!black}
 \colorlet{colorCarreUpLeft}{green!60!black}
 \colorlet{colorCarreDownRight}{yellow!50!black}
 \colorlet{colorCarreDownLeft}{yellow!90!black}
 \colorlet{colorCarreBorder}{black}
 \colorlet{colorTroisQuart}{white}
 \colorlet{colorUnQuartBas}{red!30}
 \colorlet{colorSignalCarres}{red}
 \colorlet{colorSignalQuart}{black}
 \colorlet{colorSignalBas}{cyan}

 \tikzstyle{carreUpRight}=[pattern=north west lines,pattern color=black]
 \tikzstyle{carreDownRight}=[pattern=north west lines,pattern color=black]
 \tikzstyle{carreUpLeft}=[color=colorCarreUpLeft]
 \tikzstyle{carreDownLeft}=[color=colorCarreDownLeft]
 \tikzstyle{carreBorder}=[color=colorCarreBorder]
 \tikzstyle{troisQuart}=[color=colorTroisQuart]
 \tikzstyle{unQuartBas}=[pattern=dots]
 \tikzstyle{signalCarres}=[color=colorSignalCarres]
 \tikzstyle{signalBas}=[color=colorSignalBas]
 \tikzstyle{signalQuart}=[color=colorSignalQuart]
 
 \tikzstyle{bGcarreUpRight}=[fill=colorCarreUpRight]
 \tikzstyle{bGcarreUpLeft}=[color=colorCarreUpLeft]
 \tikzstyle{bGcarreDownLeft}=[color=colorCarreDownLeft]
 \tikzstyle{bGcarreDownRight}=[color=colorCarreDownRight]
 \tikzstyle{bGcarreBorder}=[color=colorCarreBorder]
 \tikzstyle{bGtroisQuart}=[color=colorTroisQuart]
 \tikzstyle{bGunQuartBas}=[color=colorUnQuartBas]
 \tikzstyle{bGsignalCarres}=[color=colorSignalCarres]
 \tikzstyle{bGsignalBas}=[color=colorSignalBas]
 \tikzstyle{bGsignalQuart}=[color=colorSignalQuart]

\date{}

\newcommand{\ZZ}{\mathbb{Z}}

\newcommand{\NN}{\mathbb{N}}

\newcommand{\pizu}{$\Pi^0_1$\xspace}
\newcommand{\settilings}[1]{\mathcal T(#1)}
\newcommand{\turinf}{\leq_T}
\newcommand{\turequiv}{\equiv_T}

\begin{document}
\title{\pizu sets and tilings}
\author{Emmanuel Jeandel}
\email{emmanuel.jeandel@lif.univ-mrs.fr}
\author{Pascal Vanier}
\email{pascal.vanier@lif.univ-mrs.fr}
\address{Laboratoire d'Informatique Fondamentale de Marseille}

\begin{abstract}
  In this paper, we prove that given any \pizu subset $P$ of $\{0,1\}^\NN$ there
is a tileset $\tau$ with  a set of configurations $C$ such that $P\times\ZZ^2$
is recursively homeomorphic to $C\setminus U$ where $U$ is a computable set of
configurations.  As a consequence, if $P$ is countable, this tileset has the
exact same set of Turing degrees.
\end{abstract}

\maketitle

\section*{Introduction}\label{S:intro}

Wang tiles have been introduced by Wang \cite{WangII} to study fragments of 
first order logic. Knowing whether a tileset can tile the plane with a given tile at the 
origin (also known as the origin constrained domino problem) was proved undecidable 
also by Wang \cite{Wang3}. Knowing whether a tileset can tile the plane in the general case
was proved undecidable by Berger \cite{BergerPhd,Berger2}.

Understanding how complex, in the sense of recursion theory, the 
tilings of a given tileset can be is a question that was first studied by Myers  \cite{Myers} 
in 1974. Building on the work of Hanf \cite{Hanf1}, he gave a tileset with no recursive tilings.
Durand/Levin/Shen \cite{BDLS} showed, 40 years later, how to build a tileset for which all
tilings have high Kolmogorov complexity.

A \pizu-set is an effectively closed subset of $\{0,1\}^\mathbb{N}$, or
equivalently the set of oracles on which a given Turing machine halts.
\pizu-sets occur naturally in various areas in computer science and
recursive mathematics, see e.g. \cite{CenzerRec, SimpsonSurvey} 
and the upcoming book \cite{Cenzerbook}. 
It is easy to see that the set of tilings of a given tileset is a \pizu-set
(up to a recursive coding of $Q^{\mathbb{Z}^2}$ into $\{0,1\}^\mathbb{N}$).
This has various consequences. As an example, every non-empty tileset
contains a tiling which is not Turing-hard (see Durand/Levin/Shen \cite{BDLS} for a
self-contained proof). The main question is how different the sets of tilings
are from \pizu-sets. In the context of one-dimensional symbolic dynamics,
some answers to these questions were given by Cenzer/Dashti/King/Tosca/Wyman \cite{Dashti, CenzerShift,CenzerCiE}.

The main result in this direction was obtained by Simpson \cite{MEdv},
building on the work of Hanf and Myers: for
every \pizu-set $S$, there exists a tileset whose set of tilings have the
same \emph{Medvedev} degree as $S$. The Medvedev degree roughly relates to the
``easiest'' Turing degree of $S$. What we are interested in 
is a stronger result: \emph{can we find for every \pizu-set $S$ a tileset whose
set of tilings have the same Turing degrees ?} We prove in this article that
this is true if $S$ contains a recursive point. More exactly we build
(theorem \ref{thm:pizuhomtiling})  for every \pizu-set $S$ a set of tilings for which
the set of Turing degrees is exactly the same as for $S$, possibly with the
additional Turing degree of recursive points.
In particular, as every \emph{countable} \pizu-set contains a recursive
point, the question is completely solved for countable sets: the sets of
Turing degrees of countable \pizu-sets are the same as the sets of Turing
degrees of countable sets of tilings. In particular, there exist countable
sets of tilings with non-recursive points. This can be thought as a
two-dimensional version of theorem~8 in \cite{CenzerCiE}.

This paper is organized as follows. After some preliminary definitions, we start with a
quick proof of a generalization of Hanf, already implicit in Simpson
\cite{MEdv}. We then build a very specific tileset, which forms a
grid-like structure while having only countably many tilings. This tileset
will then serve as the main ingredient in the theorem in the last section.

\section{Preliminaries}\label{S:defs}
 \subsection{\pizu sets and degrees}
 A \emph{\pizu set} $P\subseteq\{0,1\}^\NN$ is a set for which there exists a Turing machine that
 given $x\in\left\{ 0,1 \right\}^\NN$ as an oracle halts if and only if $x\not\in P$. Equivalently,
a subset $S\subseteq\{0,1\}^\NN$ is \pizu if there exists a recursive set $L$ so that $w\in S$ if no 
prefix of $w$ is in $L$.

 We say that two sets $S,S'$ are \emph{recursively homeomorphic} if there exists a bijective recursive function $f:S\rightarrow S'$.

 A point $x$ of a set $S\subseteq\left\{ 0,1 \right\}^\NN$ is \emph{isolated} if it has a prefix
 that no other point of $S$ has. The \emph{Cantor-Bendixson derivative} $D(S)$ of $S$ is the set
 $S$ without its isolated points. We define inductively $S^{\left( \lambda \right)
 }$ for any ordinal $\lambda$:
 \begin{itemize}
   \item $S^{(0)}=S$
   \item $S^{(\lambda+1)}=D\left( S^{\left( \lambda\right)}) \right)$
   \item $S^{(\lambda)} = \bigcap_{\gamma<\lambda} S^{(\gamma)}$ when $\lambda$ is limit.
 \end{itemize}

 The \emph{Cantor-Bendixson rank} of $S$, noted $CB(S)$, is defined as the first ordinal $\lambda$ such that
 $S^{(\lambda)}=S^{(\lambda+1)}$. An element $x$ is of rank $\lambda$ in $S$ if $\lambda$ is the
 least ordinal such that $x\not\in S^{(\lambda)}$.
 
 See Cenzer/Remmel~\cite{CenzerRec} for \pizu sets and Kechris~\cite{Kechris} for Cantor-Bendixson rank and derivative. 
 
 For $x,y\in\{0,1\}^\NN$ we say that $x$ is \emph{Turing-reducible} to $y$ if $y$ is computable by a Turing machine using $x$ as an oracle and we write $y\turinf x$. If $x\turinf y$ and $y\turinf x$, we say that $x$ and $y$ are  \emph{Turing-equivalent} and we write $x\turequiv y$. The \emph{Turing degree} of $x\in\{0,1\}^\NN$ is its equivalence class under the relation $\turequiv$.
\newpage
 \subsection{Tilings and SFTs} 
 \emph{Wang tiles} are unit squares with colored edges which may not be flipped or rotated. A
 \emph{tileset} $T$ is a finite set of Wang tiles. A \emph{configuration} is a mapping
 $c:\ZZ^2\rightarrow T$ assigning a Wang tile to each point of the plane. If all adjacent tiles of a
 configuration have matching edges, the configuration is called a tiling. The set of all tilings
 of $T$ is noted $\settilings{T}$. We say a tileset is \emph{origin constrained} when the tile at position $(0,0)$ is forced, that is to say, we only look at the valid tilings having a given tile $t$ at 
 the origin.
 
 A \emph{Shift of Finite Type (SFT)} $X\subseteq\Sigma^{\ZZ^2}$ is defined by $(\Sigma,F)$ 
 where $\Sigma$ is a finite alphabet and $F$ a finite set of \emph{forbidden patterns}. 
 A \emph{pattern} is a coloring of a finite portion $P\subset \ZZ^2$ of the plane. A point $x$ is
 in $X$ if and only if it does not contain any forbidden pattern of $F$ anywhere. 
 In particular, the set of tilings of a Wang tileset is a SFT. Conversely, any SFT is 
 recursively homeomorphic to a Wang tileset.  More information on SFTs may be found in Lind and 
 Markus' book \cite{LindMarkus}. 

A set of configurations $X\subseteq\Sigma_X^{\ZZ^2}$ is a \emph{sofic shift} iff there exists a SFT $Y\subseteq\Sigma_Y^{\ZZ^2}$  and a local map $f:\Sigma_Y\rightarrow \Sigma_X$ such that for any point $x\in X$, there exists a point $y\in Y$ such that for all $z\in\ZZ^2, x(z)=f(y(z))$.

 The notion of \emph{Cantor-Bendixson derivative} is defined on configurations in a similar way as with
 \pizu sets. This notion was introduced for tilings by Ballier/Durand/Jeandel \cite{BDJSTACS}. A configuration $c$ is said to be \emph{isolated} in a set of configurations
 $C$ if there exists a pattern $P$ such that $c$ is the only configuration of $C$ containing $P$.
 The Cantor-Bendixson derivative of $C$ is noted $D(C)$ and consists of all configurations of $C$
 except the isolated ones. We define $C^{(\lambda)}$ inductively for any ordinal $\lambda$ as above.

 \section{\pizu sets and origin constrained tilings}\label{S:hanf}
A straighforward corollary  of Hanf \cite{Hanf1} is that \pizu subsets of $\{0,1\}^\NN $ 
and origin constrained 
tilings are recursively isomorphic. This is stated explicitely in Simpson \cite{MEdv}.
\begin{theorem}
  Given any \pizu subset $P$ of $\{0,1\}^\NN$, there exists a tileset and a tile $t$ such that 
  each origin constrained tiling with this tileset describes an element of $P$.
\end{theorem}
\begin{proof}
    We take the basic encoding of Turing machines as stated in Robinson~\cite{Robinson} for instance. 
  We modify the bottom tiles, ie the tiles containing the initial tape, such that instead of being able
to contain only the blank symbol, they can contain only 0s or 1s on the
  right of the starting head. The Turing machine we encode is the one that given $x\in
  \{0,1\}^\NN$ as an input halts if and only if $x\not\in P$. Then the constrained tilings, having at the origin the tile with the starting head of the Turing machine, are exactly the runs of the Turing
  machine on the members of $P$.  
\end{proof}
\begin{corollary}
  Any \pizu subset $P$ of $\{0,1\}^\NN$ is recursively homeomorphic to an origin constrained
  tileset.
\end{corollary}
 
\section{The tileset}\label{S:tileset}

The main problem in the construction of Hanf is that tilings which do not
have the given tile at the origin can be very wild : they may correspond to
configurations with no computation (no head of the Turing Machine) or computations starting from
an arbitrary (not initial) configuration.
A way to solve this problem is described in \cite{Myers} but is unsuitable
for our purposes.

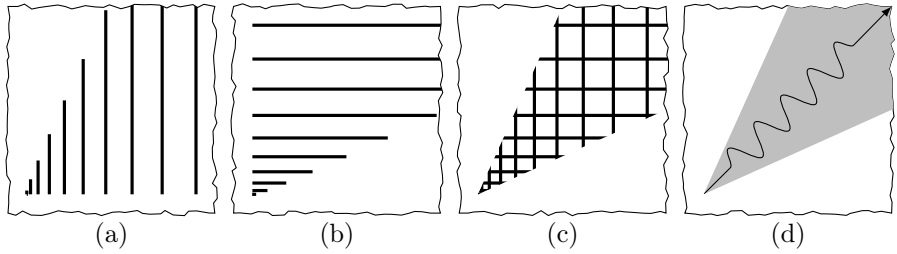
\begin{figure}[htbp]
 \begin{center}
 \begin{tikzpicture}[scale=0.025]
  \def\lx{10};
  \def\linewidth{1.2pt};
  \begin{scope}
  \begin{scope}
    \clip[draw,decorate,decoration={random steps, segment length=3pt, amplitude=1pt}] (-10,-10) rectangle (100,100);
    \foreach \i in {0,...,\lx}{
      \def\x{\i*\i-\i};
      \draw[line width=\linewidth] (\x,0) -- (\x,\i*\i*2);
    }
  \end{scope}
  \node[below] at (45,-10) {(a)};
  \end{scope}
  \begin{scope}[shift={(120,0)}]
  \begin{scope}
    \clip[draw,decorate,decoration={random steps, segment length=3pt, amplitude=1pt}] (-10,-10) rectangle (100,100);
    \foreach \i in {0,...,\lx}{
      \def\x{\i*\i-\i};
      \draw[line width=\linewidth] (0,\x) -- (\i*\i*2,\x);
    }
  \end{scope}
  \node[below] at (45,-10) {(b)};
  \end{scope}
  \begin{scope}[shift={(240,0)}]
  \begin{scope}
    \clip[draw,decorate,decoration={random steps, segment length=3pt, amplitude=1pt}] (-10,-10) rectangle (100,100);
    \begin{scope}
      \clip (0,0) -- (\lx*\lx*2,\lx*\lx-\lx) -- (\lx*\lx-\lx,\lx*\lx*2) -- cycle;
      \foreach \i in {0,...,\lx}{
        \def\x{\i*\i-\i};
        \draw[line width=\linewidth] (0,\x) -- (\i*\i*2,\x);
        \draw[line width=\linewidth] (\x,0) -- (\x,\i*\i*2);
      }
    \end{scope}
  \end{scope}
  \node[below] at (45,-10) {(c)};
  \end{scope}
   \begin{scope}[shift={(360,0)}]
  \begin{scope}[even odd rule]
    \clip[draw,decorate,decoration={random steps, segment length=3pt, amplitude=1pt}] (-10,-10) rectangle (100,100);
    \begin{scope}
      \clip (0,0) -- (\lx*\lx*2,\lx*\lx-\lx) -- (\lx*\lx-\lx,\lx*\lx*2) -- cycle;
      \fill[color=lightgray] (-10,-10) rectangle (110,110);
      \draw[-latex,decorate,decoration={snake,amplitude=2mm,segment length=5mm,pre length=5mm,post length=5mm}]
      (0,0) -- (\lx*\lx,\lx*\lx);

    \end{scope}
  \end{scope}
  \node[below] at (45,-10) {(d)};
  \end{scope}

 \end{tikzpicture}
 \end{center}
 \caption{The tiling in which to encode the Turing machines}
 \label{fig:layers}
\end{figure}

Our idea is as follows: We build a tileset which will contain, among others,
the
\emph{sparse grid} of figure~\ref{fig:layers}c. 
The main point is that all others tilings of the tileset will have at most
one intersection point of two black lines.
This means that if we put computation cells of a given
Turing machine in the intersection points, 
every tiling which is not of the form of figure~\ref{fig:layers}c
will contain at most one cell of the Turing machine, thus will contain no
computation.

\begin{figure}[htbp]
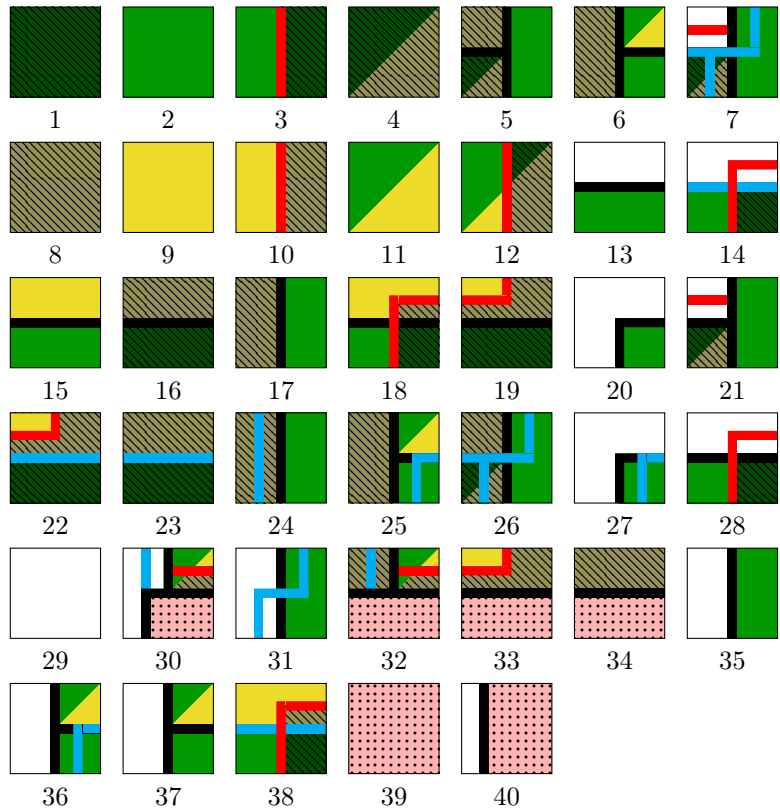

\begin{center}
\input figures/tileset.tex
\end{center}
\caption{Our set of Wang tiles $T$.}
\label{fig:tileset}
\end{figure}

To do this construction, we will first draw increasingly big
and distant columns as in figure~\ref{fig:layers}a  and then
superimposing the same construction for rows as in figure~\ref{fig:layers}b,
leading to the grid of figure~\ref{fig:layers}c.

It is then fairly straightforward to see how we can encode a Turing machine inside a configuration
having the skeleton of figure~\ref{fig:layers}c by looking at it diagonally: time increases going
to the north-east and the tape is written on the north west - south east diagonals\footnote{Note that we will have to skip one diagonal out of two in our construction, in order for the tape to increase at the same rate as the time.}.

Our set of tiles $T$ of figure~\ref{fig:tileset} gives the skeleton of figure~\ref{fig:layers}a 
when forgetting everything but the black vertical borders. We will prove in this section that it
is countable.
We set here the vocabulary:
\begin{itemize}
  \item a vertical line is formed of a vertical succession of tiles containing a vertical black line  
    (tiles 5, 6, 17, 21, 24, 25, 26, 27, 31, 35, 36, 37). 
  \item a horizontal line is formed of a horizontal succession of tiles containing a horizontal
    black line (tiles 13, 14, 15, 16, 22, 23, 38) or a bottom signal,
  \item the bottom signal
  \tikz[scale=0.25]{\fill[bGsignalBas] (0,0) rectangle (1,0.2);\fill[signalBas] (0,0) rectangle (1,0.2);}
  is formed by a connected path of tiles among
  (30, 31, 27, 14, 7, 36, 38)
  \item the red signal 
  \tikz[scale=0.25]{\fill[bGsignalCarres] (0,0) rectangle (0.2,1);\fill[signalCarres] (0,0) rectangle (0.2,1);}
  is formed by a connected path of tiles containing a red line (tiles 
    among 3 ,7, 10, 12, 14, 19, 22, 32, 33, 38).
  \item tile 30 is the corner tile
  \item tiles 30, 32, 33, 34 are the bottom tiles
\end{itemize}

\begin{lemma}\label{lem:twolines}
  The tileset $T$ admits at most one tiling with two or more vertical lines.
\end{lemma}
\begin{proof}
The idea of the construction is to force that whenever there are two
vertical lines, then the only possible tiling is the one of figure~\ref{fig:confalpha}.
Note that whenever the corner tile appears in a tiling, it is necessarily a shifted version of the tiling on figure~\ref{fig:confalpha}. 

Suppose that we have a tiling in which two vertical lines appear. Suppose they 
are at distance $k+1$. Necessarily there must be horizontal lines between them forming
squares. Inside these squares there must be a red signal: inside each square, this red
signal is vertical, it is shifted to the right each time it crosses a horizontal line.
This ensures that there are exactly $k$ squares in this column. Furthermore, the bottom square has
necessarily a bottom signal going through its top horizontal line. The bottom signal forces
the square of the column before to be of size $k-1$ and the square of the column after to be
of size exactly $k+1$.

\end{proof}

\begin{figure}
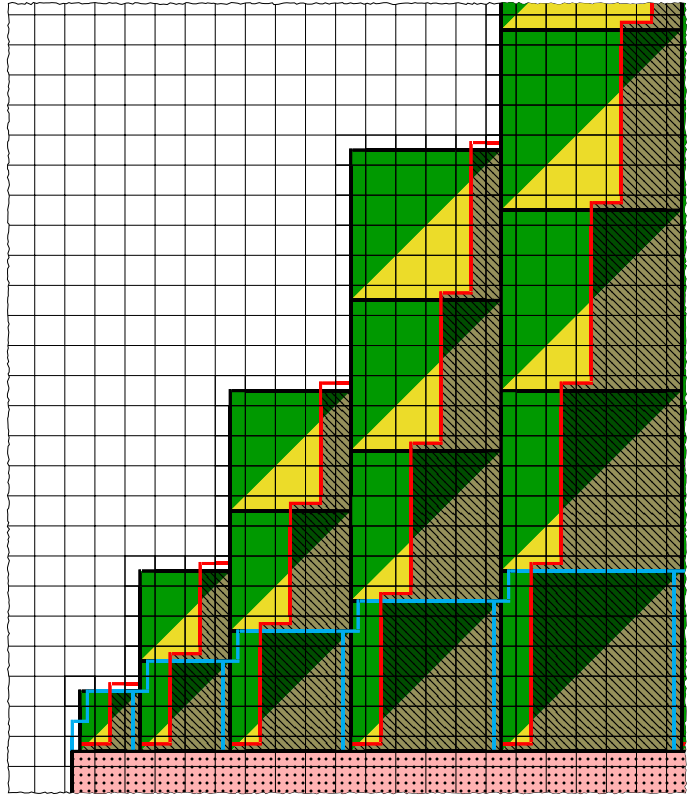
 
 \begin{center}
 \input figures/tmtiling.tex
 \end{center}
 \caption{Tiling $\alpha$: the unique valid tiling of $T$ in which there are 2 or more vertical
 lines.}
 \label{fig:confalpha}
\end{figure}

\begin{lemma}\label{lem:countable}
The tileset $T$ admits a countable number of tilings.
\end{lemma}
\begin{proof}
  Lemma~\ref{lem:twolines} states that there is only one tiling that has more than
  2 vertical lines. This means that the other tilings have at most one such line.
  \begin{itemize}
    \item If a tiling has exactly one vertical line, then it can have at most
      two horizontal lines: one on the left of the vertical one and one on the right.
      A red signal can then appear on the left or the right of the vertical line arbitrary
      far from it. There is a countable number of such tilings. 
    \item If a tiling has no vertical line, then it has at most one horizontal line.
      A red signal can then appear only once. There is a finite number of such
      tilings. 
  \end{itemize}
  There is a countable number of tilings that can be obtained with the tileset $T$.
  All obtainable tilings are shown in figure~\ref{fig:otherconf}
  and~\ref{fig:confalpha}.
  
\end{proof}

\begin{figure}
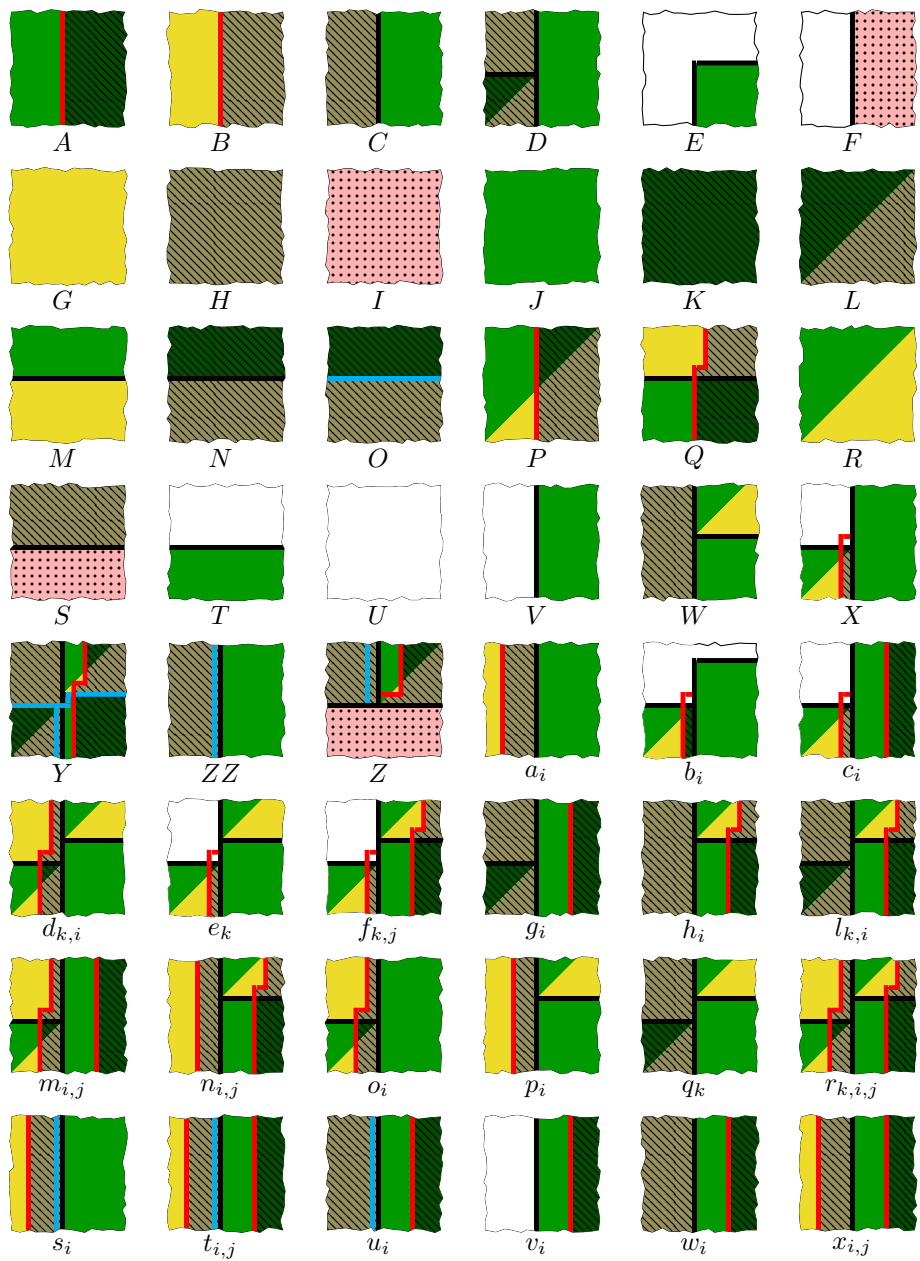

  \begin{center}
  \input figures/otherconf.tex
  \end{center}
  \caption{The other configurations: the $A-ZZ$ configurations are unique (up to
shift), and the configurations with subscripts $i,j\in\NN,k\in\ZZ^2$ represent
the fact that distances between some of the lines 
  can vary. Note that configuration $ZZ$ cannot have a red signal on its left,
because it would
  force another vertical line.}
  \label{fig:otherconf}
\end{figure}

By taking our tileset $T=\{1,\dots,40\}$ and mirroring all the tiles along the
south west-north east diagonal,
we obtain a tileset $T'=\{1',\dots,40'\}$ with the exact same properties, except
it enforces the squeleton of
figure~\ref{fig:layers}b. Remember that whenever the corner tile appeared in a
tiling, then necessarily
this tiling was $\alpha$. The same goes for $T'$ and its corner tile. 
We hence construct a third tileset $\tau =
\left(T\setminus\left\{30\right\}\times
T'\setminus\left\{30'\right\}\right)\cup\left\{(30,30')\right\}$. The corner
tile $(30,30')$ of
$\tau$ has the property that whenever it appears, the tiling is the
superimposition of the
skeletons of figures~\ref{fig:layers}a and~\ref{fig:layers}b with the corner
tiles at the same
place: there is only one such tiling, call it $\beta$. 

The skeleton of figure~\ref{fig:layers}c is obtained if we forget about the
parts of
the lines of the $T$ layer (resp. $T'$) that are superimposed to white tiles,
29' (resp. 29), 
of $T'$ (resp. $T$).

As a consequence of lemma~\ref{lem:countable}, $\tau$ is countable. And as a
consequence of
lemma~\ref{lem:twolines}, the only tiling by $\tau$ in which computation can be
embedded is $\beta$. The shape of $\beta$ is the one of
figure~\ref{fig:layers}c, the coordinates of the points of the grid are the
following (supposing tile $(30,30')$ is at the center of the grid):
 $$\left\{(f(n),f(m))\mid f(m)/4\leq f(n)\leq 4f(m)\right\}$$
 $$\left\{(f(n),f(m))\mid m/2\leq n \leq 2m\right\}$$
 
 where $f(n)=(n+1)(n+2)/2-1$.

\begin{lemma}\label{lem:rankT}
  The Cantor-Bendixson rank of $\settilings{\tau}$ is 12.
\end{lemma}
\begin{proof} 
  The Cantor-Bendixson rank of $\settilings{T}\setminus\{\alpha\}$ is 6, see
figure~\ref{fig:otherconf}, thus the rank of
  $\settilings{T}\setminus\left\{ \alpha
\right\}\times\settilings{T'}\setminus\left\{ \alpha'
  \right\}$ is 11. Adding the configurations corresponding to the
superimposition of $\alpha$ and
  $\alpha'$, $\tau$ is of rank 12.
  
\end{proof}

\section{\pizu sets and tilings}

\begin{theorem}\label{thm:pizuhomtiling}
  For any \pizu subset $S$ of $\left\{ 0,1 \right\}^\NN$ there exists a tileset
$\tau_S$ such
  that $S\times\ZZ^2$ is recursively homeomorphic to
$\settilings{\tau_S}\setminus O$ where $O$
  is a computable set of configurations.
\end{theorem} 
\begin{proof}
  This proof uses the construction of section~\ref{S:tileset}.
  Let $M$ be a Turing machine such that $M$ halts with $x$ as an oracle iff
$x\not\in S$.
  Take the tileset $\tau$ of section~\ref{S:tileset} and encode in it the Turing
machine $M$
  having as an oracle $x$ on an unmodifiable second tape.
  This gives us $\tau_M$, $O$ is the set all tilings except the $\beta$ ones. To
each $(x,p)\in
  S\times\ZZ^2$ we associate the $\beta$ tiling having a corner at position $p$
and having $x$
  on its oracle tape. It follows from lemma~\ref{lem:countable} that $O$ is
clearly computable.
  
\end{proof}
\begin{corollary}\label{corr:turingdegree}
 For any countable \pizu subset $S$ of $\left\{0,1\right\}^\NN$, there exists a
tileset $\tau$ having exactly the same Turing degrees.
\end{corollary}
\begin{proof}
 We know, from Cenzer/Remmel~\cite{CenzerRec}, that countable \pizu sets have
\textbf{0} (computable elements) in their set of Turing degrees, thus the
tileset $\tau_M$ described in the proof of theorem~\ref{thm:pizuhomtiling} has
exactly the same Turing degrees as $S$.
 
\end{proof}

\begin{theorem}\label{thm:tilerankpizu}
  For any countable \pizu subset $S$ of $\left\{ 0,1 \right\}^\NN$ there exists
a tileset $\tau_S$ such
  that $CB(\settilings{\tau_S}) = CB(S)+11$.
\end{theorem} 
\begin{proof}
 Lemma~\ref{lem:rankT} states that $\settilings{\tau}$ is of Cantor-Bendixson
rank 12, 11 without $\alpha$. In the tileset $\tau_M$ of the previous proof, the
Cantor-Bendixson rank of the contents of the tape is exactly $CB(S)$, hence 
$CB(\settilings{\tau_S})=CB(S)+11$.

\end{proof}

From Ballier/Durand/Jeandel \cite{BDJSTACS} we know that for any tileset $X$, if
$CB(\settilings{X})\geq 2$, then $X$ has only recursive points. Thus an optimal
construction
improves the Cantor-Bendixson rank by at least 2.

\begin{corollary}\label{cor:ranksofic}
  For any countable \pizu subset $S$ of $\left\{ 0,1 \right\}^\NN$ there exists
a sofic subshift $X$ such
  that $CB(X) = CB(S)+2$.
\end{corollary}
\begin{proof}
 Take a projection that just keeps the symbols of the Turing machine tape
$\tau_M$ of the proof of theorem~\ref{thm:pizuhomtiling} and maps everything
else to a blank symbol. Recall the Turing machine tape cells are the
intersections of the vertical lines and horizontal lines. This projection leads
to 
 3 possible configurations :
 \begin{itemize}
  \item a completely blank configuration,
  \item a completely blank configuration with only one symbol somewhere,
  \item a configuration with a white background and points corresponding to the
intersections in the sparse grid of figure~\ref{fig:layers}c.
 \end{itemize}

\end{proof}
\bibliographystyle{plain}
\bibliography{biblio,books}
\end{document}